\documentclass[review,english]{elsarticle}

\usepackage{lineno,hyperref}
\modulolinenumbers[5]

\journal{arXiv.org}

\usepackage[T1]{fontenc}
\usepackage[latin9]{inputenc}
\usepackage{color}
\usepackage{float}
\usepackage{amsmath}
\usepackage{graphicx}
\usepackage{amssymb}
\usepackage{epstopdf}
\makeatletter

\floatstyle{ruled}
\newfloat{algorithm}{tbp}{loa}
\providecommand{\algorithmname}{Algorithm}
\floatname{algorithm}{\protect\algorithmname}

\newtheorem{thm}{\protect\theoremname}
\newtheorem{defn}[thm]{\protect\definitionname}
\newtheorem{cor}[thm]{\protect\corollaryname}
\newtheorem{prop}[thm]{\protect\propositionname}
\newtheorem{lem}[thm]{\protect\lemmaname}
\newproof{proof}{Proof}

\makeatother

\usepackage{babel}
\providecommand{\corollaryname}{Corollary}
\providecommand{\definitionname}{Definition}
\providecommand{\lemmaname}{Lemma}
\providecommand{\propositionname}{Proposition}
\providecommand{\remarkname}{Remark}
\providecommand{\theoremname}{Theorem}

\begin{document}
\begin{frontmatter}
\title{On the Complexity of Detecting $k$-Length Negative Cost Cycles
\thanks{\textcolor{black}{This research work is supported by  Natural Science
Foundation of China \#61300025, and Australian Research Council Discovery Project DP150104871.}}
}

\author[mymainaddress]{Longkun Guo\corref{mycorrespondingauthor}}
\cortext[mycorrespondingauthor]{Corresponding author}
\ead{lkguo@fzu.edu.cn}

\author[mysecondaryaddress]{Peng Li}
\ead{lipeng.net@gmail.com}

\address[mymainaddress]{College of Mathematics and Computer Science,
Fuzhou University, China}
\address[mysecondaryaddress]{Amazon Web Services, Amazon.com, Unite States}


\begin{abstract}
Given a positive integer $k$ and a directed graph with a cost on each edge, the $k$-length negative cost cycle ($k$\emph{LNCC})
problem is to determine whether there exists a negative cost cycle
with at least $k$ edges, and the fixed-point \emph{$k$-}length negative
cost cycle \emph{trail (FP$k$LNCCT)} problem is to determine whether
there exists a negative trail enrouting a given vertex (as the fixed
point) and containing only cycles with at least $k$ edges. The $k$\emph{LNCC}
problem first emerged in deadlock avoidance in synchronized streaming
computing network \cite{spaa10}, generalizing two famous problems: negative cycle
detection and the $k$-cycle problem. As a warmup by-production, the paper
first shows that \emph{FP$k$LNCCT is }${\cal NP}$-complete in multigraph\emph{
}even for\emph{ $k=3$} by reducing from the \emph{3SAT} problem.
Then as the main result, we prove the ${\cal NP}$-completeness of
$k$\emph{LNCC} by giving a sophisticated reduction from the 3 Occurrence
3-Satisfiability (\emph{3O3SAT}) problem, a known ${\cal NP}$-complete
special case of 3SAT in which a variable occurs at most three times.
The complexity result is interesting, since polynomial time algorithms
are known for both $2$\emph{LNCC} (essentially no restriction on
the value of $k$) and the
$k$-cycle problem of fixed $k$. This paper closes the open
problem proposed by Li et al. in \cite{spaa10} whether $k$\emph{LNCC}
admits polynomial-time algorithms.
\end{abstract}


\end{frontmatter}

\section{Introduction}

We define the following \textbf{\emph{$k$}}\emph{-length negative
cost cycle problem ($k$LNCC):}
\begin{defn}
\emph{Given a fixed integer $k$ and a directed graph $G=(V,\,E)$,
in which each edge $e\in E$ is with a cost $c(e)\rightarrow\mathbb{R}$
and a length $l(e)=1$, }\textbf{\emph{$k$}}\emph{LNCC is to determine
whether there exists a cycle $O$ with total length $l(O)=\sum_{e\in O}l(e)\geq k$
and total cost $c(O)<0$.}
\end{defn}
The \emph{$k$LNCC} problem arises in deadlock avoidance for streaming
computing, which is widely used in realtime analytics, machine learning,
robotics, and computational biology, etc. A streaming computing system
consists of networked nodes communicating through finite first-in
first-out (FIFO) channels, and a data stream referring to data transmitted
through a channel. In streaming computing, a compute node might need
to synchronize different incoming data streams. If the synchronized
streams have different rates (e.g. due to filtering \cite{spaa10}),
the computing network might deadlock. Several deadlock avoidance algorithms
have been proposed in \cite{spaa10}, which rely on inserting heartbeat
messages into data streams. When to insert those heartbeat messages,
however, depends on the network topology and buffer size configurations.
An open problem is deciding whether a given heartbeat message schedule
can guarantee deadlock freedom, which raises the negative-cost cycle
detection problem. Further, we are only interested in the negative-cost
cycle with length at least $k\geq3$, as deadlocks of few nodes can
be easily eliminated. This raises the above\emph{ $k$LNCC} problem.
Besides, in many cases, we are also interested in whether a particular
node in a streaming computing system is involved in a deadlock or
not, which brings the \textbf{\emph{$k$-}}\emph{length negative-cost
fixed-point cycle trail problem (FP$k$LNCCT)}:
\begin{defn}
\emph{Given a graph $G=(V,\,E)$ in which each edge $e\in E$ is with
a cost $c(e)\rightarrow\mathbb{R}$ and a length $l(e)=1$, a fixed
integer $k$, and}\textbf{\emph{ }}\emph{a fixed point $p\in G$,
the}\textbf{\emph{ }}\emph{FP$k$LNCCT problem is to determine whether
there exists a trail $T$ containing $p$, such that $c(T)<0$ and
$l(O)=\sum_{e\in O}l(e)\geq k$ for every $O\subseteq T$ (with $c(O)<0$).}
\end{defn}
Note that $c(O)<0$ makes no difference for the above definition,
as a negative cost trail must contain at least a negative cost cycle.
Further, if $G$ contains a vertex $p$ as a fixed point such that
the\textbf{ }FP$k$LNCCT problem is feasible, then $G$ must contain
at least a cycle\emph{ $O$ }with\emph{ $c(O)<0$ }and\emph{ $\sum_{e\in O}l(e)\geq k$,
}and\emph{ vice versa. }So if \emph{FP}$k$\emph{LNCCT} admits a polynomial-time
algorithm, then $k$\emph{LNCC }is also\emph{ }polynomially solvable.
That is because we can run the polynomial-time algorithm as a subroutine
to verify whether $G$ contains a vertex $p$ wrt which \emph{FP}$k$\emph{LNCCT}
is feasible, and then to verify whether $k$\emph{LNCC} is feasible.
Conversely, if the ${\cal NP}$-completeness of $k$\emph{LNCC}\textbf{
}is proven, it can be immediately concluded that \emph{FP}$k$\emph{LNCCT}\textbf{\emph{
}}is also ${\cal NP}$-complete\textbf{\emph{. }}

Throughout this paper, by \emph{walk }we mean an alternating sequence
of vertices and connecting edges; b\emph{y trail }we mean a walk that
does not pass over the same edge twice;\emph{ }by\emph{ path }we mean
a \emph{trail} that does not include any vertex twice; and by\emph{
cycle }we mean a path that begins and ends on the same vertex.

\subsection{Related works}

The $k$\emph{LNCC} problem generalizes two well known problems: the
negative cycle detection problem of determining whether there exist
negative cycles in a given graph, and the $k$-cycle problem (or namely
the long directed cycle problem \cite{cygan2015parameterized}) of
determining whether there exists a cycle with at least $k$ edges.
The former problem is known polynomially solvable via the Bellman-Ford
algorithm \cite{bellman1956routing,ford1956network} and is actually
$k$\emph{LNCC} of $k=2$. The latter problem, to determine whether
a given graph contains a cycle $O$ with $l(O)\geq k$, is $k$\emph{LNCC}
when $c(e)=-1$ for every $e\in E$. It is shown fixed parameter tractable
in \cite{gabow2008finding}, where an algorithm with a time complexity
$k^{O(k)}n^{O(1)}$ is proposed. The runtime is then improved to $O(c^{k}n^{O(1)})$
for a constant $c>0$ by using representative sets \cite{fomin2014efficient},
and later to $6.75^{k}n^{O(1)}$ independently by \cite{fomin2014representative}
and \cite{shachnai2014faster}. Compared to the two above results,
i.e. negative cycle detection and the $k$-cycle problem of fixed
$k$ are both polynomially solvable, it is interesting that \emph{$3$LNCC}
is ${\cal NP}$-complete as \emph{$3$LNCC }is exactly a combination
of the two problems belonging to ${\cal P}$.

\subsection{Our results}

The main result of this paper is proving the ${\cal NP}$-completeness
of $k$\emph{LNCC} in a simple directed graph. Since the proof is
constructive and complicated, we will first accomplish a much easier
task of proving the ${\cal NP}$-completeness of the \emph{FP}$k$\emph{LNCCT}
problem in multigraphs by simply reducing from the 3-Satisfiability
(\emph{3SAT}) problem, where a multigraph is a graph that allows multiple
edges between two nodes.
\begin{lem}
\label{lem:MG}For any fixed integer $k\geq3$, FP$k$LNCCT is ${\cal NP}$-complete
in a multigraph.
\end{lem}
Then following a similar main idea of the proof of Lemma \ref{lem:MG}
but with more sophisticated details, we will prove the ${\cal NP}$-completeness
of $k$\emph{LNCC} (and hence also\emph{ FP}$k$\emph{LNCCT}) in a
simple directed graph.
\begin{thm}
\label{thm:1} For any fixed integer $k\geq3$, $k$LNCC is ${\cal NP}$-complete
in a simple graph.
\end{thm}

\section{The ${\cal NP}$-completeness of \emph{FP}$k$\emph{LNCCT} in Multigraphs}

In this section, we will prove Lemma \ref{lem:MG} by reducing from
\emph{3SAT} that is known to be ${\cal NP}$-complete \cite{gary1979computers}.
\textcolor{black}{In an instance of }\textcolor{black}{\emph{3SAT}}\textcolor{black}{,
we are given $n$ variables $\{x_{1},\dots,x_{n}\}$ and $m$ clauses
$\{C_{1},\dots,C_{m}\}$, where $C_{i}$ is the }\textcolor{black}{\emph{OR}}\textcolor{black}{{}
of at most three }\textcolor{black}{\emph{literals}}\textcolor{black}{,
and each literal is an occurrence of the variable $x_{j}$ or its
negation. The }\textcolor{black}{\emph{3SAT}}\textcolor{black}{{} problem
is to determine whether there is an assignment satisfying all the
$m$ clauses. }

For any given instance of \emph{3SAT}, the key idea of our reduction
is to construct a digraph $G$, such that $G$ contains a negative
cost \emph{trail} with only cycles of length at least 3 and enrouting
a given vertex iff the instance of \emph{3SAT} is satisfiable. The
construction is composed with the following three parts. First, for
each variable $x_{i}$ with $a_{i}$ occurrences of $x_{i}$ and $b_{i}$
occurrences of $\overline{x}_{i}$ in the clauses, we construct a
\emph{lobe}\footnote{The term \emph{lobe }was used to denote an unit of the auxiliary graph
constructed for an instance of SAT, as in \cite{guo2013finding} and
many others \cite{xu2006caa,fortune1980directed}. } which contains two vertices, denoted as $y_{i}$ and $z_{i}$, and
$a_{i}+b_{i}$ edges of cost $-1$ between the two vertices, i.e.
$a_{i}$ copies of edge $(y_{i},\,z_{i})$ and $b_{i}$ copies of
$(z_{i},\,y_{i})$ ( A lobe is depicted as in the dashed circles in
Figure \ref{fig:The-construction-of-1}). For briefness, we say an
edge in the lobes is a\emph{ lobe-edge}. Then, for each clause $C_{j}$,
add two vertices $u_{i}$ and $v_{i}$, as well as edge $(v_{i},u_{i+1})$,
$1\leq i\leq m-1$, with cost $0$ and edge $(v_{m},\,u_{1})$ with
cost $m-1$. Last but not the least, for the relationship between
the variables and the clauses, say variable $x_{j}$ occurs in clause
$C_{i}$, we add two edges with cost $0$ to connect the lobes and
the vertices of clauses:
\begin{itemize}
\item If $C_{i}$ contains $x_{j}$, then add two edges $(u_{i},\,y_{j}),\,(z_{j},\,v_{i})$;
\item If $C_{i}$ contains $\overline{x}_{j}$, then add two edges $(u_{i},\,z_{j}),\,(y_{j},\,v_{i})$.
\end{itemize}
An example of the construction for a \emph{3SAT} instance is depicted
in Figure \ref{fig:The-construction-of-1}.

\begin{figure}
\includegraphics[width=0.95\columnwidth]{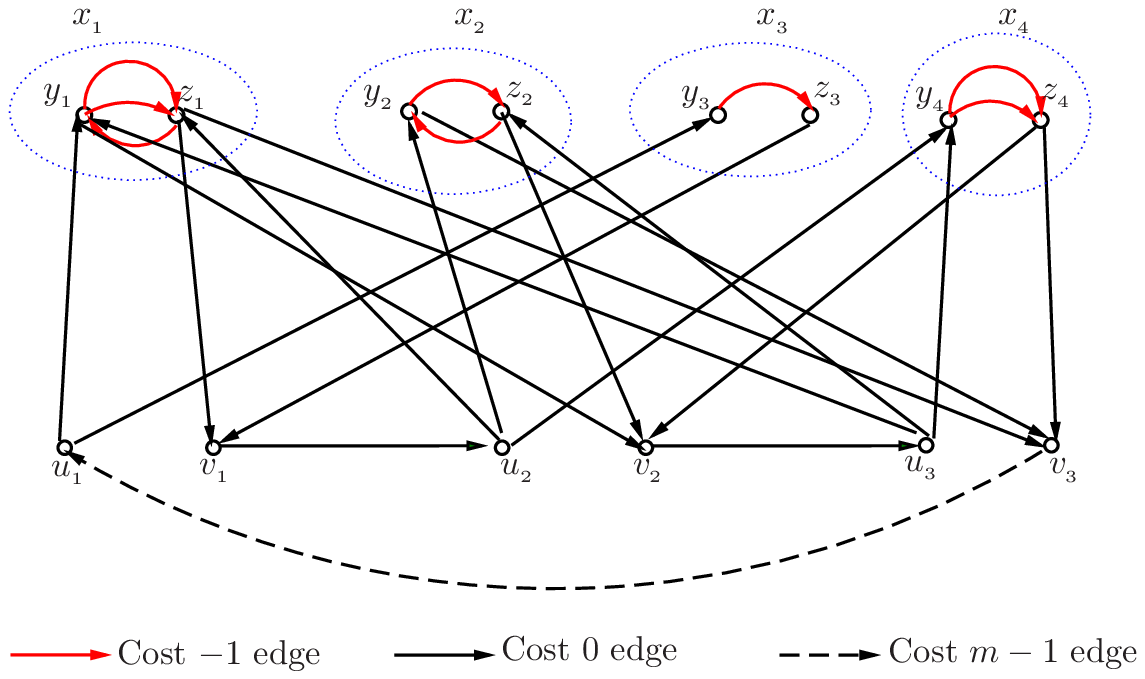}

\caption{\label{fig:The-construction-of-1}The construction of $G$ for an
\emph{3SAT} instance $(x_{1}\vee x_{3})\land(\overline{x}_{1}\vee x_{2}\vee x_{4})\land(x_{1}\vee\overline{x}_{2}\vee x_{4})$.}
\end{figure}

Then since \emph{FP$k$LNCCT} is clearly in ${\cal NP}$, the correctness
of Lemma \ref{lem:MG} can be immediately obtained from the following
lemma:
\begin{lem}
\label{lem:3sattolnfcc}An instance of 3SAT is satisfiable iff in
its corresponding auxiliary graph $G$ there exists a negative-cost
trail containing $u_{1}$ but no length-2 cycles.
\end{lem}
\begin{proof}
Suppose there exists a negative-cost trail $T$, which contains $u_{1}$
but \emph{NO} negative cost length-2 cycle. Then let $\tau$ be a
true assignment for the \emph{3SAT} instance according to $T$: if
$T$ goes through $(y_{i},\,z_{i})$, then set $\tau(x_{i})=true$;
Otherwise, set $\tau(x_{i})=false$. It remains to show such an assignment
will satisfy all the clauses. Firstly, we show that the path $P=T\setminus\{v_{m},\,u_{1}\}$
must go through all vertices of $\{v_{i}\vert i=1,\,\dots,\,m\}$.
Since $T$ contains $u_{1}$, $T$ has to go through edge $(v_{m},\,u_{1})$,
as the edge is the only one entering $u_{1}$. Then because $T$ is
with negative cost and the cost of edge $(v_{m},\,u_{1})$ is $m-1$,
$P$ has to go through at least $m$ edges within the $n$ lobes,
as only the edges of lobes has a negative cost $-1$. According to
the construction of $G$, between two lobe-edges on $P$, there must
exist at least an edge of $\{v_{i},\,u_{i+1}\vert i=1,\,\dots,\,m-1\}$,
since $v_{i}$ has only one out-going edge $(v_{i},\,u_{i+1})$ while
every edge leaving a lobe must enter a vertex of $\{v_{i}\vert i=1,\,\dots,,\,m-1\}$.
So $P=T\setminus\{v_{m},\,u_{1}\}$ has to go through all the $m-1$
edges of $\{v_{i},\,u_{i+1}\vert i=1,\,\dots,\,m-1\}$, and hence
through all vertices of $\{v_{i}\vert i=1,\,\dots,\,m\}$. Secondly,
assume $v_{j}$ and $v_{j'}$ are two vertices of $\{v_{i}\vert i=1,\,\dots,\,m\}$,
such that $P(v_{j},\,v_{j'})\cap\{v_{i}\vert i=1,\,\dots,\,m\}=\{v_{j},\,v_{j'}\}$.
Again, because there must be at least an edge of $\{v_{i},\,u_{i+1}\vert i=1,\,\dots,\,m-1\}$
between two lobe-edges, there must be at least a lobe-edge appearing
on $P(v_{j},\,v_{j'})$, otherwise there will be at most $m-1$ lobe-edges
on $T$. That is, there must be exactly a lobe edge, say $(y_{i},\,z_{i})$,
on $P(v_{j},\,v_{j'})$. Then according to the construction of graph
$G$, $x_{i}$ appears in $C_{j}$, and hence $\tau(x_{i})=true$
satisfies $C_{j}$. The case for $(z_{i},\,y_{i})$ appears on $P(v_{j},\,v_{j'})$
is similar. Therefore, the \emph{3SAT} instance is feasible as it
can be satisfied by $\tau$.

Conversely, suppose the instance of \emph{3SAT} is satisfiable, and
a true assignment is $\tau:\,x\rightarrow\{true,\,false\}$. Then
for clause $C_{k}$, there must exist a literal, say $w_{k}$, such
that $\tau(w_{k})=true$. If $w_{k}$ is an occurrence of $x_{i}$,
then set the corresponding subpath as $P_{k}=u_{k}-y_{i}-z_{i}-v_{k}$;
otherwise set $P_{k}=u_{k}-z_{i}-y_{i}-v_{k}$. Then clearly, $P=\{P_{k}\vert k=1,\dots,\,m\}\cup\{(v_{h},\,u_{h+1})\vert h=1,\dots,\,m-1\}$
exactly composes a path from $u_{1}$ to $v_{m}$ with a cost of $-m$,
as it contains $m$ lobe-edge and other edges of cost 0. So $T=P\cup\{v_{m},\,u_{1}\}$
is a negative cost trail of length at least 3. Besides, since $\tau(x_{i})$
must be either \emph{true} or \emph{false}, there exist no length-2
cycles on $P$. This completes the proof.\qed
\end{proof}
However, the above proof can not be immediately extended to prove
the ${\cal NP}$-completeness of $k$\emph{LNCC}, since there are
two tricky obstacles. First,\emph{ }in the above proof, there might
exist negative cycles with length at least three but without going
through $u_{1}$. Thus, in Lemma \ref{lem:3sattolnfcc}, containing
$u_{1}$ is mandatory. we Second, Lemma \ref{lem:3sattolnfcc} holds
only for multigraphs as some of the lobes are already multigraphs.
Thus, to extend the proof to $k$\emph{LNCC}, we need to eliminate
negative cycles bypassing $u_{1}$ and to transform the (multigraph)
lobes to simple graphs.

\section{The ${\cal NP}$-completeness Proof of $k$\emph{LNCC} }

In this section, to avoid the two obstacles as analyzed in the last
section, we will prove Theorem \ref{thm:1} by reducing from the \textcolor{black}{3
occurrence }\textcolor{black}{\emph{3SAT}} (\emph{3O3SAT}) problem
that is known ${\cal NP}$-complete \cite{tovey1984simplified}. \textcolor{black}{Similar
to }\textcolor{black}{\emph{3SAT,}}\textcolor{black}{{} in an instance
of }\textcolor{black}{\emph{3O3SAT}}\textcolor{black}{{} we are also
given $m$ clauses $\{C_{1},\dots,C_{m}\}$ and $n$ variables $\{x_{1},\dots,x_{n}\}$,
and the task is to determine whether there is an assignment satisfying
all the $m$ clauses. The only difference is, however, each variable
$x_{i}$ (including both literal $x_{i}$ and $\overline{x}_{i}$)
appears at most 3 times in all the $m$ clauses. }To simplify the
reduction, we assume that the possible occurrences of a variable $x$
in an instance of \emph{3O3SAT} fall in the following three cases:
\begin{itemize}
\item \textbf{Case 1:} All occurrences of $x$ are all positive literal
$x$;
\item \textbf{Case 2:} The occurrences of $x$ are exactly one positive
literal and one negative literal\textcolor{black}{; }
\item \textbf{Case 3:} The occurrences of $x$ are exactly two positive
literals and one negative literal.
\end{itemize}
The above assumption is without loss of generality. We note that there
are still two other cases:
\begin{itemize}
\item \textbf{Case 4:} All occurrences of $x$ are negative literals;
\item \textbf{Case 5:} Exactly two occurrences of negative literals and
one positive literal.
\end{itemize}
However, Case 4 and Case 5 can be respectively reduced to Case 1 and
Case 3, by replacing occurrences of $\overline{x}$ and $x$ respectively
with $y$ and $\overline{y}$. Therefore, we need only to consider\emph{
3O3SAT} instances with variables satisfying Case 1-3.

The key idea of the proof is, for any given instance of \emph{3O3SAT},
to construct a graph $G$, such that there exists a cycle $O$ with
$c(O)<0$ and $l(O)\geq3$ in $G$ if and only if the instance is
satisfiable.  An important fact used in the construction is that
every variable appears at most 3 times in a \emph{3O3SAT} instance.
In the following, we will show how to construct $G$ according to
clauses, variables, and the relation between clauses and variables.
\begin{enumerate}
\item For each $C_{k}$:

Add to $G$ two vertices $u_{k}$ and $v_{k}$, as well as edge $(v_{k},u_{k+1})$,
$1\leq k\leq m-1$ with cost 0, and edge $(v_{m},u_{1})$ with cost
$-1$.
\item For each variable $x_{i}$, construct a lobe according to the occurrences
of $x_{i}$ and $\overline{x}_{i}$ (The construction a lobe is as
depicted in a dashed circle as in Figure \ref{fig:The-construction-of}):

\begin{itemize}
\item \textbf{Case 1:} All occurrences of $x_{i}$ in are positive literal
$x_{i}$, such as $x_{4}$ in Figure \ref{fig:The-construction-of}.

For the $j$th occurrence of $x_{i}$, add a directed edge $(y_{i}^{j},z_{i}^{j})$
and assign cost $-2m$ to it.
\item \textbf{Case 2:} Exactly one occurrence for each of positive literal
$x_{i}$ and negation $\overline{x}_{i}$, such as $x_{2}$ in Figure
\ref{fig:The-construction-of}.

\begin{enumerate}
\item Add two vertices $z_{i}^{j_{2}}=y_{i}^{j_{1}}$ and $y_{i}^{j_{2}}=z_{i}^{j_{1}}$,
and connect them with directed edges $(y_{i}^{j_{1}},z_{i}^{j_{1}})$
and $(y_{i}^{j_{2}},z_{i}^{j_{2}})$.
\item Assign edge $(y_{i}^{j_{1}},z_{i}^{j_{1}})$ with cost $-2m$ and
$(y_{i}^{j_{2}},z_{i}^{j_{2}})$ with cost $\frac{1}{m+1}$;
\end{enumerate}
\item \textbf{Case 3: }Exactly 2 occurrences of $x_{i}$ and one occurrence
$\overline{x}_{i}$, such as $x_{1}$ in Figure \ref{fig:The-construction-of}.

\begin{enumerate}
\item For the two positive literals of $x_{i}$, say the $j_{1}$th and
$j_{2}$th occurrence of $x_{i}$, $j_{1}<j_{2}$, add four vertices
$y_{i}^{j_{1}},\,z_{i}^{j_{1}}$, $y_{i}^{j_{2}},\,z_{i}^{j_{2}}$,
and two directed edges $(y_{i}^{j_{1}},z_{i}^{j_{1}})$, $(y_{i}^{j_{2}},z_{i}^{j_{2}})$
connecting them with cost $-2m$;
\item For the negation $\overline{x}_{i}$, say the $j_{3}$th occurrence,
set $z_{i}^{j_{3}}=y_{i}^{j_{1}}$ and $y_{i}^{j_{3}}=z_{i}^{j_{2}}$,
and add three directed edges $(y_{i}^{j_{3}},\,y_{i}^{j_{2}}),\,(y_{i}^{j_{2}},\,z_{i}^{j_{1}}),\,(z_{i}^{j_{1}},\,z_{i}^{j_{3}})$
with costs $\frac{1}{2m+2}$, 0 and $\frac{1}{2m+2}$, respectively.
\end{enumerate}
\end{itemize}
\item For the relation between the variables and the clauses, say $C_{k}$
is the clause containing the $j$th occurrence of $x_{i}$, i.e. $C_{k}$
is the $j$th clause $x_{i}$ appears in, add directed edges $(u_{k},y_{i}^{j})$
and $(z_{i}^{j},v_{k})$. If the occurrence of $x_{i}$ in $C_{k}$
is a positive literal, assign the newly added edges with cost $m$;
Otherwise, assign them with cost 0. Note that no edges will be added
between lobes and $u_{k},\,v_{k}$ if $x_{i}$ does not appear in
$C_{k}$.
\end{enumerate}
An example of the construction of $G$ according to $F=(x_{1}\vee x_{3})(\overline{x}_{1}\vee x_{2}\vee x_{4})(x_{1}\vee\overline{x}_{2}\vee x_{4})$
is as depicted in Figure \ref{fig:The-construction-of}.

\begin{figure}
\includegraphics[width=0.95\columnwidth]{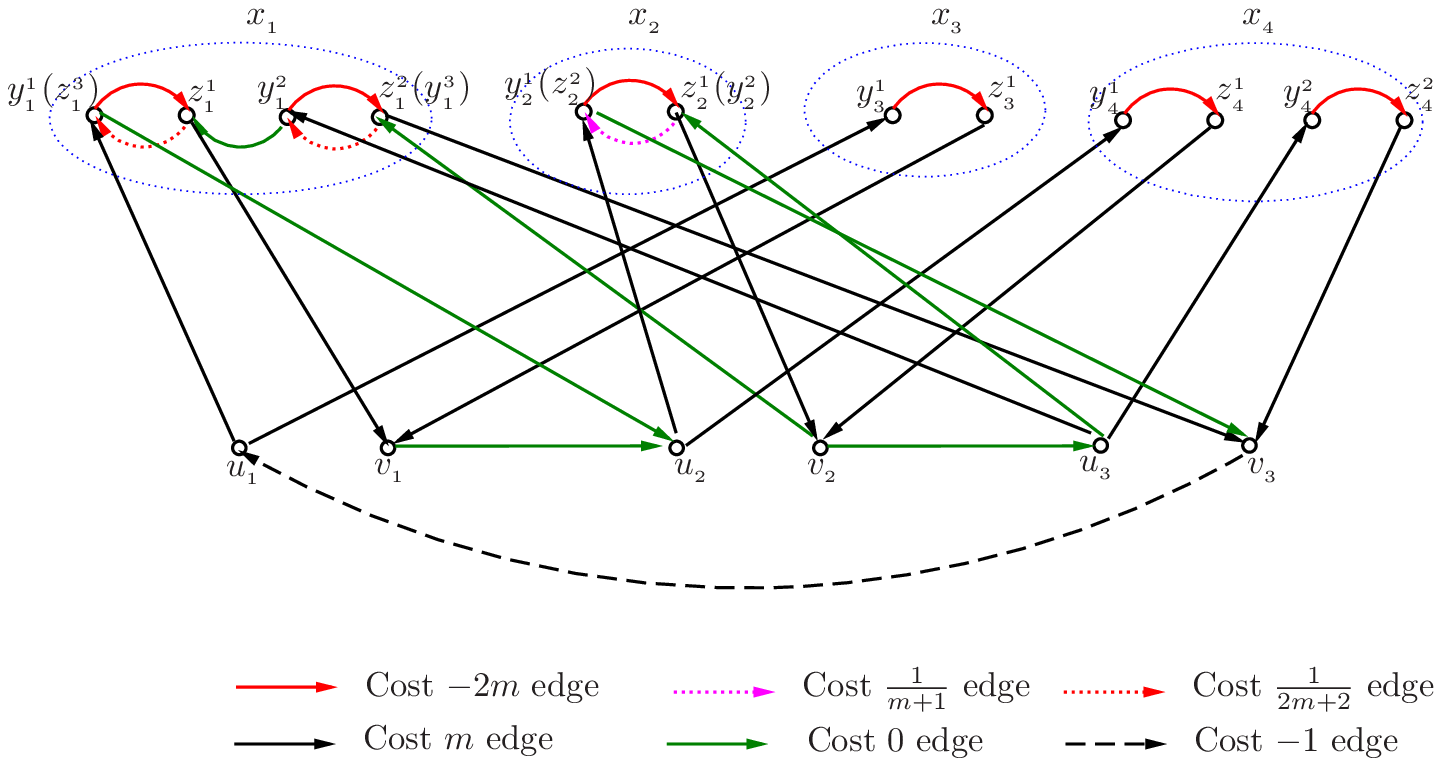}

\caption{\label{fig:The-construction-of}The construction of $G$ for an \emph{3O3SAT}
instance $(x_{1}\vee x_{3})\land(\overline{x}_{1}\vee x_{2}\vee x_{4})\land(x_{1}\vee\overline{x}_{2}\vee x_{4})$.}
\end{figure}

As \emph{$k$LNCC} is apparently in ${\cal NP}$, it remains only
to prove the following lemma to complete the proof of Theorem \ref{thm:1}.
\begin{lem}
\label{lem:reduction}An instance of $3O3SAT$ is feasible iff the
corresponding graph $G$ contains a cycle $O$ with $l(O)\geq3$ and
$c(O)<0$.
\end{lem}
Let $U=\{u_{i},\,v_{i}\vert i=1,\dots,m\}$ be the set of vertices
that correspond to the clauses. We first prove a proposition that
if a path from $u_{h}$ to $v_{l}$ ($h\neq l$) does not enroute
any other $u\in U$, the cost of the path is at least $m$.
\begin{prop}
\label{prop:neighbour}Let $P(u,\,v)$ be a path from $u$ to $v$.
For any path $P(u_{h},\,v_{l})$ that satisfies$P(u_{h},\,v_{l})\cap U=\{u_{h},\,v_{l}\}$,
if $h\neq l$, then $c(p(u_{h},\,v_{l}))\geq m$.
\end{prop}
\begin{proof}
For every edge $(y,\,z)$ with cost $-2m$, clearly there exists only
one edge $e_{1}$ entering $y$, and only one edge $e_{2}$ leaving
$z$. Furthermore, $e_{1}=(u_{h'},\,y)$ and $e_{2}=(z,\,v_{h'})$.
So $P(u_{h},\,v_{l})$, $h\neq l$, as in the proposition can not
go through any cost $-2m$ edge. That is, $P(u_{h},\,v_{l})$, $h\neq l$,
can only go through the non-negative cost edges. It remains to show
$P(u_{h},\,v_{l})$, $h\neq l$, must go through at least one cost
$m$ edge.

Suppose $P(u_{h},v_{l}),\,h\neq l$ does not go through any cost $m$
edges. Let $(u_{h},\,y_{i}^{j_{1}}),$$\,(z_{i}^{j_{2}},\,v_{l})$$\in P(u_{h},\,v_{l})$
be the two edges leaving $u_{h}$ and entering $v_{l}$, respectively.
Then $y_{i}^{j_{1}}$ and $z_{i}^{j_{2}}$ must incident to two edges
that corresponds to the negation of two variables. Further, the two
variable must be identical, since vertices of two distinct lobes will
be separated by $U$, and hence for $P(y_{i}^{j_{1}},\,z_{i}^{j_{2}})\subset P(u_{h},\,v_{l})$,
$P(y_{i}^{j_{1}},\,z_{i}^{j_{2}})\cap U\neq\emptyset$. This contradicts
with $P(u_{h},\,v_{l})\cap U=\{u_{h},\,v_{l}\}$.\qed
\end{proof}
\begin{prop}
\label{prop:nonnegative}In graph $G\setminus e(v_{m},\,u_{1})$,
every path $P(u,\,v)$, $u,\,v\in U$, has a non-negative cost.
\end{prop}
\begin{proof}
Apparently, in $G\setminus e(v_{m},\,u_{1})$, every edge with negative
cost is exactly with cost $-2m$. Let $(y_{i}^{j},\,z_{i}^{j})$ be
such an edge with cost $-2m$. From the structure of $G$, there exists
exactly one edge entering $y_{i}^{j}$, and exactly one leaving $z_{i}^{j}$,
each of which is with exactly cost $m$. So for every path $P(u,\,v)$,
$u,\,v\in U$, if it contains edge $(y_{i}^{j},\,z_{i}^{j})$, then
it must also go through both the edge entering $y_{i}^{j}$ and the
edge leaving $z_{i}^{j}$. That is, the three edges must all present
or all absent in $P(u,\,v)$, and contribute a total cost 0. Therefore,
$c(P(u,\,v))\geq0$ must hold.\qed
\end{proof}
Now the proof of Lemma \ref{lem:reduction} is as below:
\begin{proof}
Suppose that there exists an assignment $\tau:\,x\rightarrow\{true,\,false\}$
satisfying all the $m$ clauses. Since $c(v_{m},\,u_{1})=-1$, we
need only to show there exists a $u_{1}v_{m}$-path with cost smaller
than 1 by construction such one path. For a satisfied clause $C_{k}$,
there must exist a literal, say $w_{k}$ with $\tau(w_{k})=true$.
If $w_{k}$ is the $j$th occurrence of $x_{i}$, then set the corresponding
subpath as $P_{k}=u_{k}-y_{i}^{j}-z_{i}^{j}-v_{k}$. We need only
to show $P=\{P_{k}\vert k=1,\dots,\,m\}\cup\{(v_{h},\,u_{h+1})\vert h=1,\dots,\,m-1\}$
exactly composes a path from $u_{1}$ to $v_{m}$ with cost smaller
than 1. For the first, $P$ is a path. Because $\tau$ is an assignment,
$\tau(x_{i})=true$ and $\tau(\overline{x}_{i})=true$ can not both
hold, and hence $P$ contains no length-2 cycle. For the cost, according
to the construction, if $\tau(w_{k})=\tau(x_{i})=true$ then the subpath
$P_{k}$ is with cost exactly equal to 0; otherwise, i.e. $\tau(w_{k})=\tau(\overline{x}_{i})=true$,
the subpath $P_{k}$ is with cost exactly equal to $\frac{1}{m+1}$.
Meanwhile, $c(e(v_{h},\,u_{h+1}))=0$ for each $h$. Therefore, the
total cost $c(P)\leq\frac{m}{m+1}<1$, where the maximum is attained
when all clauses are all satisfied by negative of the variables.

Conversely, assume that there exists a negative-cost cycle $O$, which
contains NO negative cost length-2 cycle. According to Proposition
\ref{prop:nonnegative},  $e(v_{m},u_{1})$ must appears on $O$,
so that $c(O)<0$ can hold. Let $P=O\setminus e(v_{m},u_{1})$ and
$\tau$ be a true assignment according to $P$: if $P$ goes through
literal $\overline{x}_{i}$, set $\tau(x_{i})=false$ and set $\tau(x_{i})=true$
otherwise. It remains to show such the assignment according to $P$
satisfies all the clauses. To do this, we shall firstly show $P$
will go through all the vertices of $U$ in the order $u_{1}\prec v_{1}\prec\dots\prec u_{i}\prec v_{i}\prec\dots\prec u_{m}\prec v_{m}$;
and secondly show that $P(u_{h},\,v_{h})$ has to go through exactly
a subpath corresponding to a literal, say $w$, for which if $\tau(w)=true$,
then $C_{h}$ is satisfied. Then from the fact that $P$ contains
no negative cost length-2 cycle, $\tau$ is a feasible assignment
satisfying all the clauses.\textbf{}

For the first, according to Proposition \ref{prop:neighbour}, if
$P(u_{i},\,v_{j})\cap U=\{u_{i},\,v_{j}\}$, then $j=i$ (i.e. $v_{j}=v_{i}$)
must hold. Since otherwise, according to Proposition \ref{prop:neighbour}
$c(P(u_{i},\,v_{j}))\geq m$ must hold; while according to Proposition
\ref{prop:nonnegative}, the other parts of $P$ is with $c(P(u_{1},\,u_{i}))\geq0$
and $c(P(v_{j},\,v_{m}))\geq0$. That is, $c(P)\geq m$. On the other
hand, since $c(e(v_{m},\,u_{1}))=-1$ and $c(O)<0$, we have $c(P)<1$,
a contradiction. Furthermore, since there exists only one edge leaving
$v_{i}$, i.e. $(v_{i},\,u_{i+1})$, $P$ must go through every edge
$e(v_{i},\,u_{i+1})$ incrementally on $i$, i.e. in the order of
$u_{1}\prec v_{1}\prec\dots\prec u_{i}\prec v_{i}\prec\dots\prec u_{m}\prec v_{m}$.
For the second, according to the structure of $G$ and $c(P)<1$,
$c(P(u_{h},\,v_{h}))\leq\frac{1}{m+1}$ must hold. Then $c(P(u_{h},\,v_{h}))$
has to go through exactly a subpath corresponding to a literal.\qed
\end{proof}
Note that, a simple undirected graph does not allow length-2 cycles.
Anyhow, by replacing length-2 cycles with length-3 cycles in the above
proof, i.e. replacing every edge $(y_{j},\,z_{j})$ with two edges
$(y_{j},\,w_{j})$ and $(w_{j},\,z_{j})$ of the same cost, and ignoring
the direction of the edges, we immediately have the correctness of
Corollary \ref{cor:simplek4}.
\begin{cor}
\label{cor:simplek4}For any fixed integer $k\geq4$, $k$LNCC is
${\cal NP}$-complete in a simple undirected graph.
\end{cor}

\section{Conclusion}

In this paper, we have shown the ${\cal NP}$-completeness for both
the $k$-length negative cost\textbf{ }cycle ($k$\emph{LNCC}) problem
in a simple directed graph and the fixed-point \textbf{\emph{$k$-}}length
negative cost cycle \emph{(FP$k$LNCCT)} problem in a directed multigraph,
which have wide applications in parallel computing, particularly in
deadlock avoidance for streaming computing systems. Consequently,
it can be concluded that $k$LNCC is ${\cal NP}$-complete in a simple
undirected graph and that \emph{FP$k$LNCCT} is also ${\cal NP}$-complete
in a simple directed graph. In future, we will investigate approximation
algorithms for these two problems.

\bibliographystyle{plain}
\bibliography{kcycle}

\end{document}